\documentclass[10pt]{article}

\usepackage[latin9]{inputenc}
\usepackage{latexsym, bm, amsmath, amssymb, graphics, amsthm,
enumerate} %
\usepackage[a4paper,
            bindingoffset=0.2in,
            left=0.8in,
            right=0.8in,
            top=0.8in,
            bottom=0.8in,
            footskip=.25in]{geometry}
\usepackage{xcolor,placeins}
\usepackage[ruled,vlined]{algorithm2e}
\usepackage[round]{natbib}
\usepackage{tikz}

\usetikzlibrary{calc,trees,positioning,arrows,chains,shapes.geometric,  decorations.pathreplacing,decorations.pathmorphing,shapes,  matrix,shapes.symbols}
\usepackage{tikz-cd}
\tikzcdset{scale cd/.style={every label/.append style={scale=#1},
    cells={nodes={scale=#1}}}}

\usepackage{hyperref} 
\usepackage{multirow,multicol}

\usetikzlibrary{arrows.meta,
                positioning,
                shapes}
\newtheorem{theorem}{Theorem}

\newtheorem{Def}{Definition}

\newtheorem{Ex}{Example}

\usepackage{comment}

\usepackage{booktabs}

\begin{document}

\title{Delayed Arrow-of-Time Detection in Signed Laplacian Dynamics\footnote{We thank Stuart Zoble for very helpful conversations.  Financial support from NYU Stern School of Business, NYU Shanghai, J.P. Valles, and the HHL - Leipzig Graduate School of Management is gratefully acknowledged.  ChatGPT o3 and SuperGrok were used to help find proof tactics, create numerical examples, and identify references.}}

\author
{Adam Brandenburger \footnote{Stern School of Business, Tandon School of Engineering, NYU Shanghai, New York University, New York, NY 10012, U.S.A., adam.brandenburger@stern.nyu.edu}
\and{Pierfrancesco La Mura \footnote{HHL - Leipzig Graduate School of Management, 04109 Leipzig, Germany, plamura@hhl.de}}
    }
\date{{\textcolor{blue}{Preliminary Version}}  \\ \vspace{3mm} \today}
\maketitle

\thispagestyle{empty}

\begin{abstract}
We study how rapidly the direction of time becomes operationally detectable from mesoscopic data when state-weights may be positive or negative.  In contrast with classical Markov processes -- where forward evolution is instantly distinguishable from its reverse -- signed dynamics can render the arrow of time undetectable during an initial interval.  Assume the generator of the signed dynamics is a symmetric signed Laplacian with a single linear invariant and a phase-space Second Law holds in the form of non-decreasing R\'enyi-$2$ entropy.  Drawing on recent results on eventual exponential positivity of signed Laplacians (Chen et al., 2021), we define a test that correctly identifies the direction of thermodynamic time if conducted over a time interval of length at least $\tau$.  We go on to prove that the test cannot deliver an incorrect conclusion if conducted over a shorter interval.  Dropping symmetry, we exhibit a superquantum example where $\tau = +\infty$, so that the arrow of time remains permanently undetectable under our test.
\end{abstract}

\section{Introduction}
Detecting the temporal arrow from intrinsic system data is straightforward for classical Markov processes.  Strict positivity of forward propagators arises immediately, while the backward propagators permanently contain at least one negative entry.  In quantum phase space, the same diagnostic fails.  Signed quasi-probability propagators may begin with mixed signs in both temporal directions, making the arrow, at least initially, undetectable.

This paper demonstrates that initial undetectability is indeed possible, and we also identify conditions under which it is transient.  We develop a test on a signed-probability system that -- after an initial inconclusive interval -- discloses the arrow-of-time.  Our test operates at the mesoscopic level by preparing and measuring phase-space state distributions.  We distinguish this from the microscopic level, where the experimenter has or obtains exact information about the (infinitesimal) generator.  Our question also differs from macroscopic-level explanations of the arrow of time, such as arguments from thermodynamics (Boltzmann, 1872) or cosmology (Carroll and Chen, 2005).  Indeed, we will assume a Second Law in the form of non-decreasing R\'enyi-$2$ entropy on phase space.  Our question is if and how an experimenter can determine the direction of thermodynamic time from measurement and observation at the mesoscopic level on a signed-probability (e.g., quantum) system.

Evolution is driven by (the negative of) a symmetric signed Laplacian $\Lambda$, i.e., a symmetric real matrix whose rows sum to zero.  The only linear invariant is conservation of total probability ($\operatorname{corank}(\Lambda) = 1$).  Forward and backward propagators are $F(t) = e^{t\Lambda}$ and $B(t) =e^{-t\Lambda}$.  Our sign test is passed at a time $t$ if one propagator is strictly positive and the other contains at least one negative entry.  Otherwise, the outcome of the test is deemed inconclusive.

We adopt the Second Law axiom that the R\'enyi-$2$ entropy of every evolving signed probability vector is non-decreasing.  Applying the recent eventual exponential positivity theorem of Chen et al.~(2021), we show the existence of a finite time $0 < \tau < +\infty$ such that $F(t)$ is strictly positive for all $t \ge \tau$.  An auxiliary argument shows that $B(t)$ necessarily retains at least one negative entry.  Accordingly, every sign test conducted over an interval of length~$\tau$ or longer identifies the thermodynamic arrow of time without error.

Suppose this test is conducted at a time $t < \tau$ and returns a strictly positive backward propagator and a forward propagator with a negative entry.  The experimenter will then misidentify the direction of time.  We prove that this cannot happen, so that the test is reliable.  We go on to study to the asymmetric case ($\Lambda^{\!\top} \not= \Lambda$) and produce an example where the detection interval $\tau = +\infty$, so that the experimenter can never determine the direction of time.  We also comment on a converse direction when we assume finite-time detectability.

Our work is complementary to recent work on the quantum arrow of time, notably Rubino, Manzano, and Brukner (2021) and Guff, Shastry, and Rocco (2025).  The first paper shows how the arrow of time may become undefined through superposition of forward and backward processes, and how measurement restores the arrow.  The second paper argues that, contrary to common belief, the quantum Langevin and Lindblad master equations are time-reversal symmetric, not asymmetric, when derived in a neutral way.  Our focus is complementary in asking about the operational accessibility of an assumed quantum thermodynamic arrow of time.

The rest of the paper is organized as follows.  Section \ref{framework} formalizes signed-probability dynamics, our R\'enyi Second Law, and the recent matrix theory we use (Noutsos, 2006; Noutsos and Tsatsomeros, 2008; Olesky, Tsatsomeros, and van den Driessche, 2009; Chen et al., 2021).  Section \ref{arrow} formulates our arrow-of-time test, proves that it works (at time $\tau$ onward), and gives an example to demonstrate delayed detection.  Section \ref{reliability} proves that the test cannot deliver a false result if conducted before time $\tau$.  Section \ref{asymmetric} presents an example with infinitely delayed detection.

\section{Framework} \label{framework}
We consider dynamics on a finite phase space $X = \{x_1, \dots, x_n\}$ for an $n$-state system, with probabilities on states described by real-valued column vectors:
\begin{equation}
p(t) = (p_1(t), \dots, p_n(t))^{\!\top} \in \mathbb{R}^n.
\end{equation}
Probabilities are non-negative in the classical case but may be negative (and, possibly, greater than $1$) in the signed case.  We let $\Lambda \in \mathbb{R}^{n \times n}$ be the generator of a continuous-time Markov process, as follows.  Given an initial -- possibly signed -- probability measure $p(0) = p_0 \in \mathbb{R}^n$, we write:
\begin{equation}
p(t) = e^{t \Lambda} \, p_0 \,\, \text{for all} \,\, t \ge 0.
\end{equation}

We now specify the dynamics.  Set $\Lambda = - L$.

\begin{Def} \label{def1}
A generator $\Lambda$ is \textbf{signed Laplacian} if $L$ is a corank-$1$ signed Laplacian $L$, that is: (i) $L^{\!\top} = L$; (ii) $L \mathbf{1} = \mathbf{0}$; and (iii) $\operatorname{corank}(L) = 1$.
\end{Def}

Condition (i) is the microreversibility assumption of statistical mechanics.  Condition (ii) is conservation of total probability.  Condition (iii) says this is the sole linear invariant of the system.  The Laplacian is signed because we do not require the off-diagonal entries $L_{ij}$ to be non-positive.

Next, we introduce our entropy measure.  We choose R\'enyi $2$-entropy, since it is well behaved in both classical and signed-probability regimes.

\begin{Def} \label{def2}
The \textbf{R\'enyi $2$-entropy} of an unsigned or signed probability measure $p$ is given by:
\begin{equation}
H_2(p) = - \log_2 \bigl(\sum_{i=1}^n p_i^2 \bigr).
\end{equation}
\end{Def}

This is a standard definition in the classical (unsigned) case.  It is equally well defined in the signed case and can be derived axiomatically for signed measures via an extension of the original R\'enyi (1961) and Dar\'oczy (1963) axiomatizations for unsigned measures.  See Brandenburger and La Mura (2025), where it is shown that R\'enyi $2$-entropy -- unlike a signed Shannon entropy formula -- is well-behaved on signed phase space.  (It witnesses probabilistic negativity, is Schur-concave, satisfies a quantum H-theorem, and is constant under discrete Moyal-bracket evolution.)
\vspace{0.1in}

\noindent \textbf{Axiom (R\'enyi-$2$ Second Law).}  \textit{For every trajectory $\bigl(p(t)\bigr)_{t \ge 0}$:}
\begin{equation}
\frac{d}{dt} \, H_2\bigl(p(t)\bigr) \geq 0 \,\, \text{\textit{for all}} \,\, t \ge 0.
\end{equation}

\begin{theorem} \label{thm1}
The Second Law implies that $\Lambda$ is negative semidefinite.
\end{theorem}

\begin{proof}
The Second Law implies:
\begin{equation}
\frac{d}{dt} \, H_2\bigl(p(t)\bigr) = - \frac{2 p(t)^{\!\top} \Lambda \,p(t)}{\log_e 2 \, \|p(t)\|_2^2} \ge 0,
\end{equation}
from which:
\begin{equation} \label{eq5}
x^{\!\top}\Lambda \, x \le 0,
\end{equation}
for all $x \in \mathbb{R}^n$ with $\sum_i x_i = 1$.  Next, write any $y \in \mathbb{R}^n$ as:
\begin{equation}
y = \mu\mathbf 1 + z \,\, \text{where} \,\, \mu = \frac{1}{n}\mathbf{1}^{\!\top} y \,\, \text{and} \,\, \mathbf{1}^{\!\top}z = 0.
\end{equation}
Set:
\begin{equation} \label{eq7}
x = z + \frac{1}{n}\mathbf{1} = y + \bigl( \frac{1}{n} - \mu \bigr)\mathbf{1},
\end{equation}
so that $\mathbf{1}^\mathsf{T}x = 1$.  From Equation \ref{eq7}, we calculate:
\begin{equation}
y^{\!\top} \Lambda \, y = x^{\!\top} \Lambda \, x - x^{\!\top} \Lambda \, (\frac{1}{n} - \mu) \mathbf{1} - (\frac{1}{n} - \mu)\mathbf{1}^{\!\top} \Lambda \, x + (\frac{1}{n} - \mu)\mathbf{1}^{\!\top} \Lambda \, (\frac{1}{n} - \mu)\mathbf{1}.
\end{equation}
Using $\Lambda \mathbf{1} = 0$, and therefore $\mathbf{1}^\mathsf{T} \Lambda^\mathsf{T} = \mathbf{1}^\mathsf{T} \Lambda = 0$, all terms on the right side vanish except the first.  Since $\mathbf{1}^\mathsf{T}x = 1$, we can use Equation \ref{eq5} to conclude:
\begin{equation}
y^{\!\top} \Lambda \, y \le 0,
\end{equation}
as required.
\end{proof}

Next is the recent matrix theory we will use.  We quote a definition (Noutsos and Tsatsomeros, 2008) and result (Chen at al., 2021, Theorem 6; Noutsos, 2006; and Olesky, Tsatsomeros, and van den Driessche, 2009).  Given a matrix $M$, define $\tau$ by:
\begin{equation}
\tau = \inf \{t > 0 : e^{sM} > 0 \,\, \text{for all} \,\, s \ge t\}.
\end{equation}
That is, time $\tau$ is the first time at which $e^{sM}$ has all strictly positive entries for all times $s \ge \tau$.  If there is no such $\tau$, we write $\tau = +\infty$.

\begin{Def} \label{def3}
A matrix $M$ is \textbf{eventually exponentially positive} if $\tau < +\infty$.
\end{Def}

\begin{theorem} \label{thm2}
Fix a signed Laplacian $M$.  Then $M$ is positive semidefinite if and only if $-M$ is eventually exponentially positive.
\end{theorem}

We note for later in the paper that Theorem 6 in Chen at al.~(2021) is actually stronger.  They do not assume that $M$ has corank $1$ and prove: The matrix $M$ is positive semidefinite and has corank $1$ if and only if $-M$ is eventually exponentially positive.

\section{Arrow-of-Time Test} \label{arrow}
We now develop what we will call the arrow-of-time test to ask if an experimenter can determine which direction of the system under study coincides with ``objective" thermodynamic time given by the Second Law.  At laboratory time $t = 0$ (the zero is obviously arbitrary), the experimenter tomographically prepares $n$ independent initial signed probability measures on phase space $X$, which we write as:
\begin{equation}
{\cal S} = [p^{(1)}(0) \; \cdots \; p^{(n)}(0)], 
\end{equation}
where $\text{rank}\, S = n$.  Note that, to constitute a basis, some $p^{(k)}(0)$ must contain negative components.  Each initial measure evolves under the generator $\Lambda$ until a later laboratory time $t > 0$.  Measurements record the $n$ output states, which we write as:
\begin{equation}
{\cal O} = [p^{(1)}(t) \; \cdots \; p^{(n)}(t)]. 
\end{equation}
These measurements are to be understood as obtained via ensemble-based tomography, or as weak or minimally disturbing.

The next task is for the experimenter to fit a forward propagator $F \in \mathbb{R}^{n \times n}$ to the data, that is, to find $F$ such that ${\cal O} = F \, {\cal S}$.  By invertibility of $\cal S$, the unique solution is $F = e^{t\Lambda}$.

We can suppose that the experimenter simply inverts $F$ to obtain a candidate backward propagator.  In fact, this is the unique backward propagator.  Let the experimenter solve ${\cal S} = B \, {\cal O}$.  Since the product of two invertible matrices is invertible, we know that ${\cal O} = e^{t\Lambda} \, {\cal S}$ is invertible, from which $B = {\cal S} \, {\cal O}^{-1} = {\cal S} \, {\cal S}^{-1} \, e^{-t\Lambda} = e^{-t\Lambda}$.

In sum, we see that our preparation and measurement scheme leads the experimenter to identify unique forward and backward propagators.  This leads to our test.

\begin{quote}
\textbf{Arrow-of-Time (AoT) Test:} \textit{There is a finite time $\tau < +\infty$ such that if the experimenter identifies forward and backward propagators over the time interval $[0, t]$, for any $t \ge \tau$, then the direction of thermodynamic time is conclusively determined.}
\end{quote}

We now apply the matrix theory from the previous section to analyze this test.

\begin{theorem} \label{thm3}
There is a finite time $0 < \tau < +\infty$ such that for all times $t \ge \tau$, the forward propagator $e^{t\Lambda}$ is strictly positive.
\end{theorem}

\begin{proof}
Since $\Lambda$ is a signed Laplacian, so is $M = - \Lambda$.  The matrices $\Lambda$ and $M$ have the same null spaces, so $\operatorname{corank}(M) = 1$.  Since $\Lambda$ is negative semidefinite, it follows that $M$ is positive semidefinite.  By the ``only if" direction of Theorem \ref{thm2}, we conclude that $e^{-tM} = e^{t\Lambda}$ is eventually exponentially positive.
\end{proof}

\begin{theorem} \label{thm4}
For any time $t \ge \tau$, the backward propagator (i.e., inverse) $e^{-t\Lambda}$ contains at least one negative entry.
\end{theorem}

\begin{proof}
Suppose, to the contrary, that $B(t) = e^{-t\Lambda} \ge 0$.  From $FB = \mathbb{I}$, we know that for any indices $i \not = j$:
\begin{equation}
0 = (FB)_{ij} = \sum_{k=1}^n F_{ik}\,B_{kj}.
\end{equation}
Since every $F_{ik} > 0$, we must have $B_{kj} = 0$ for every $k$.  Now consider the diagonal entry:
\begin{equation}
1 = (FB)_{jj} = \sum_{k=1}^n F_{jk}\,B_{kj} = 0,
\end{equation}
a contradiction.
\end{proof}

We see that the AoT test works if conducted over a time interval of length at least $\tau$.  One propagator will be found to have all strictly positive entries and the other propagator to have at least one negative entry.  The two propagators are distinguishable and the first one is aligned with the thermodynamic arrow of time -- i.e., the direction in which entropy is non-decreasing.  We next give an example to show that the test can indeed be inconclusive if conducted over too short a time interval.

\begin{Ex} \label{ex1}
Consider the $4 \times 4$ generator:
\begin{equation}
\Lambda
  =\frac13
   \begin{pmatrix}
    -7 & -1 &  2 &  6\\
    -1 & -7 &  2 &  6\\
     2 &  2 & -10&  6\\
     6 &  6 &  6 & -18
   \end{pmatrix}.
\end{equation}
The spectrum is $\{0, -2, -4, -8\}$, which immediately establishes that $\operatorname{corank}(\Lambda) = 1$ and $\Lambda$ is negative semidefinite.  We set $F(t) = e^{t\Lambda}$ and $B(t) = e^{-t\Lambda}$.  Table \ref{tab:signed-qubit-propagators} shows the entry-wise extrema of the propagators at two time points.

\begin{table}[h]
\centering
\begin{tabular}{@{}cccccc@{}}
\toprule
\makebox[2em]{$t$} &
$\displaystyle\min F(t)$ &
$\displaystyle\max F(t)$ &
$\displaystyle\min B(t)$ &
$\displaystyle\max B(t)$ &
\textbf{Test verdict} \\ \midrule
0.05 & $-0.010$ & $+0.895$ & $-0.123$ & $+1.369$ & inconclusive \\
0.20 & $+0.007$ & $+0.677$ & $-0.988$ & $+3.965$ & conclusive \\ \bottomrule
\end{tabular}
\caption{Extremal Entries of the Forward and Backward Propagators}
\label{tab:signed-qubit-propagators}
\end{table}
We see that at time $t = 0.05$, the AoT test is inconclusive, because the forward propagator contains a negative entry.  The crossover occurs at $\tau \simeq 0.17$.  At $t = 0.20$, the AoT test is passed, because the forward propagator is strictly positive while the backward propagator contains a negative entry.  The experimenter correctly identifies the direction of thermodynamic time.
\end{Ex}

Finally, in this section, we recall that, by standard arguments, there is an immediately reliable AoT test for classical systems.  Fix an irreducible unsigned Laplacian $L$ and set $\Lambda = - L$.  Then, by Perron-Frobenius theory, the forward propagator $e^{t\Lambda}$ is strictly positive for all $t > 0$,   Our Theorem \ref{thm3} applies (for any $t > 0$) to conclude that the backward propagator $e^{-t\Lambda}$ contains at least one negative entry.  As expected, the test is reliable instantaneously, i.e., $\tau = 0^+$.

\section{Test Reliability} \label{reliability}
So far, we have not specified the epistemic state of the experimenter.  If the experimenter knows the generator $\Lambda$, then $\tau$ can be calculated (at least up to bounds).  We assume that the generator $\Lambda$ is not directly observable, reflecting the practical challenge of characterizing microscopic dynamics.  Instead, we suppose that the experimenter can obtain mesoscopic information.  Through preparation and measurement of states, the experimenter can calculate the propagators $F(t)$ and $B(t)$ and then conducts the AoT test.  Could the test deliver a misleading result?  Specifically, is it possible that, at some time $t < \tau$, the backward propagator has all strictly positive entries and the forward propagator has at least one negative entry.  The following result rules this out.

\begin{theorem} \label{thm5}
For any $t > 0$, the backward propagator $e^{-t\Lambda}$ contains at least one negative entry.
\end{theorem}

\begin{proof}
We actually prove the stronger result that each row of $e^{-t\Lambda}$ contains at least one negative entry.  Since $\Lambda\mathbf 1 = \mathbf 0$, the Taylor series of the matrix exponential yields:
\begin{equation}
B(t)\mathbf 1 = \Bigl[\sum_{k=0}^{\infty} \frac{(-t\Lambda)^k}{k!}\Bigr]\mathbf 1 = \mathbf 1, \label{eq18}
\end{equation}
that is, every row of $B(t)$ sums to $1$.  Since $\Lambda$ is real and symmetric, the spectral theorem gives us an orthonormal eigenbasis $\{u_k\}_{k=1}^n$:
\begin{equation}
\Lambda = U\, D \, U^{\!\top},
\end{equation}
where:
\begin{equation}
U = \bigl[u_1 \, u_2 \, \ldots \, u_n\bigr] \in \mathbb R^{n\times n}, \,\, \text{and} \,\, D = \operatorname{diag}(\lambda_1, \dots, \lambda_n).
\end{equation}

Since $\Lambda\mathbf 1 = \mathbf 0$, the constant vector $\frac1{\sqrt n}\mathbf 1$ is an eigenvector (which we label $u_1$) with eigenvalue $0$ (which we label $\lambda_1)$.  All other eigenvalues $\lambda_k < 0$ by assumption.  It follows that:
\begin{equation}
B(t) = e^{-t\Lambda} = U \, e^{-tD} \, U^{\!\top} = U \, \operatorname{diag}\bigl(1, e^{-t\lambda_2}, \dots, e^{-t\lambda_n}\bigr) \, U^{\!\top},
\end{equation}
from which:
\begin{equation} \label{eq22}
\bigl(B(t)\bigr)_{ij} = \frac{1}{n} + \sum_{k=2}^n e^{-t\lambda_k} \, u_{ki} \, u_{kj}.
\end{equation}
Fix a row $i$.  From orthonormality:
\begin{equation} \label{eq23}
\sum_{k=2}^n u_{k,i}^2 =  1 - u_{1, i}^2 = 1 - \frac{1}{n}.
\end{equation}
Set $j = i$ in Equation \ref{eq22} to obtain:
\begin{equation}
\bigl(B(t)\bigr)_{ii} = \frac{1}{n} + \sum_{k=2}^n e^{-t\lambda_k} \, u_{ki}^2 >  \frac{1}{n} + \sum_{k=2}^n u_{ki}^2 =  \frac{1}{n} + 1 -  \frac{1}{n} = 1,
\end{equation}
using $e^{-t\lambda_k} > 1$ since each $\lambda_k < 0$ for $k \ge 2$, and Equation \ref{eq23}.

But row $i$ sums to $1$ (Equation \ref{eq18}).  Since the diagonal entry is greater than $1$, we conclude that the row has at least one negative off-diagonal entry.
\end{proof}

We see that the experimenter does not need to know $\tau$.  The AoT test will be inconclusive and will not deliver a misleading result if it is conducted early -- i.e., at a time $t < \tau$.  The test is fully operational in this sense.

\section{The Asymmetric Case} \label{asymmetric}
If we drop symmetry, the AoT test may be inconclusive for all finite time intervals $[0, t]$.  Theorem \ref{thm3} fails.

\begin{Ex} \label{ex2}
Consider the $3 \times 3$ generator:
\begin{equation}
\Lambda=\begin{pmatrix}
 0 &  1 & -1\\
-1 &  0 &  1\\
 1 & -1 &  0
\end{pmatrix}.
\end{equation}
It satisfies all properties of Definition \ref{def1} except symmetry.  The Second Law holds (everywhere weakly).  The matrix $\Lambda$ is a generator of rotation about the axis $\mathbf 1$ with angular speed $\sqrt3$, specifically:
\begin{equation}
F(t) = e^{t\Lambda} = R\bigl(\theta(t)\bigr), \,\, \theta(t) = \sqrt{3}\,t,
\end{equation}
where the rotation matrix is:
\begin{equation}
R(\theta) = \frac{1}{3}
\begin{pmatrix}
1+2\cos\theta & 1-\cos\theta+\sqrt3\sin\theta & 1-\cos\theta-\sqrt3\sin\theta\\
1-\cos\theta-\sqrt3\sin\theta & 1+2\cos\theta & 1-\cos\theta+\sqrt3\sin\theta\\
1-\cos\theta+\sqrt3\sin\theta & 1-\cos\theta-\sqrt3\sin\theta & 1+2\cos\theta
\end{pmatrix}.
\end{equation}
The backward propagator is the opposite rotation $B(t) = \bigl(F(t)\bigr)^{-1} = R(-\theta)$.  Assume $R(\theta)$ is strictly positive and focus on the first row:
\begin{align} \label{eq27}
1 + 2\cos\theta &> 0, \\
1 - \cos\theta + \sqrt3\sin\theta &> 0, \label{eq28} \\
1 - \cos\theta - \sqrt3\sin\theta &> 0. \label{eq29}
\end{align}
Inequalities \ref{eq28} and \ref{eq29} together imply:
\begin{equation}
1 - \cos\theta > \sqrt3 \, |\sin\theta\,|.
\end{equation}
Since both sides are non-negative, we can square them, use $\sin^2\theta + \cos^2\theta = 1$, and rearrange to obtain:
\begin{equation}
4\cos^2\theta - 2\cos\theta - 2 > 0,
\end{equation}
from which $\cos\theta < - \frac{1}{2}$ or $\cos\theta > 1$ (impossible).  But Inequality \ref{eq27} says that $\cos\theta > - \frac{1}{2}$, a contradiction.  We conclude that there is no finite time $t > 0$ at which $F(t)$ is strictly positive, that is: $\tau = +\infty$.  Thus, the AoT test fails for this example.  An experimenter conducting the test will never be able to determine the arrow of time.

We note that this example also shows that Theorem \ref{thm5} fails without symmetry.  It can be checked that at each time $t_k= \frac{2\pi k}{3\sqrt3}$, for $k \in \mathbb{Z}$, the rotation angle is $\frac{2\pi k}{3}$ and the backward propagator is the permutation matrix that cyclically shifts coordinates -- i.e., $B(t)$ contains all $0$'s or $1$'s.
\end{Ex}

The generator in this example is antisymmetric: $\Lambda^{\!\top} = - \Lambda$.  Physically, the example describes a circulating, divergence-free signed-probability current with no dissipative pull toward the classical simplex.  In fact, any corank-$1$ generator $\Lambda$ that arises as a phase-space representation of Lindbladian dynamics has finite $\tau$, so the example is superquantum.

The converse direction of Theorem 6 in Chen et al.~(2021) gives us a converse here: If $\Lambda$ is symmetric, has row sums $0$, and passes the AoT test, then it has a Lindbladian realization.  It also has corank $1$ and is negative semidefinite -- and therefore satisfies the Second Law.

\section*{References}

\noindent Boltzmann, L., ``Weitere Studien \"uber das W\"ormegleichgewicht unter Gasmolek\"ulen," \textit{Sitz.-Ber.~Kais.~Akad. Wiss.~Wien (II)}, 66, 1872, 275-370.
\vspace{0.1in}

\noindent Brandenburger, A., and P. La Mura, ``Axiomatization of R\'enyi Entropy on Quantum Phase Space," 2025, at https://arxiv.org/abs/2410.15976.
\vspace{0.1in}

\noindent Carroll, S., and J. Chen, ``Does Inflation Provide Natural Initial Conditions for the Universe," \textit{General Relativity and Gravitation}, 37, 2005, 1671-1674.
\vspace{0.1in}

\noindent Chen, W., D. Wang, J. Liu, Y. Chen, S.Z. Khong, T. Ba\c sar, K. Johansson, and L. Qiu, ``On Spectral Properties of Signed Laplacians with Connections to Eventual Positivity," \textit{IEEE Transactions on Automatic Control}, 66, 2021, 2177-2190.
\vspace{0.1in}

\noindent Dar\'oczy, Z., ``\"Uber die gemeinsame Charakterisierung der zu den nicht vollst\"andig en Verteilungen geh\"origen Entropien von Shannon und von R\'enyi," \textit{Zeitschrift f\"ur Wahrscheinlichkeitstheorie und verwandte Gebiete}, 1, 1963, 381-388.
\vspace{0.1in}

\noindent Guff, T., C.U. Shastry, and A. Rocco, ``Emergence of Opposing Arrows of Time in Open Quantum Systems," \textit{Scientific Reports}, 15, 2025, 3658.
\vspace{0.1in}

\noindent Noutsos, D., ``On Perron-Frobenius Property of Matrices Having Some Negative Entries," \textit{Linear Algebra and Its Applications}, 412, 2006, 132-153.
\vspace{0.1in}

\noindent Noutsos, D., and M. Tsatsomeros, ``Reachability and Holdability of Nonnegative States," \textit{ SIAM Journal on Matrix Analysis and Applications}, 30, 2008, 700-712.
\vspace{0.1in}

\noindent Olesky, D., M. Tsatsomeros, and P. van den Driessche, ``$\mathrm{M}_{\vee}$-Matrices: A Generalization of $\mathrm{M}$-Matrices Based on Eventually Nonnegative Matrices," \textit{The Electronic Journal of Linear Algebra}, 18, 2009, 339-351.
\vspace{0.1in}

\noindent R\'enyi, A., ``On Measures of Information and Entropy," in Neyman, J. (ed.), \textit{Proceedings of the 4th Berkeley Symposium on Mathematical Statistics and Probability}, University of California Press, 1961, 547-561.
\vspace{0.1in}

\noindent Rubino, G., G. Manzano, and \u C. Brukner, ``Quantum Superposition of Thermodynamic Evolutions with Opposing Time's Arrows," \textit{Communications Physics}, 4, 2021, 245.

\end{document}